\documentclass[authoryear,a4paper, 12pt]{elsarticle}
\usepackage{natbib}
\usepackage{amsthm}
\usepackage{amssymb}
\usepackage{graphicx,amsmath}
\usepackage{tikz}
\usepackage{mathtools}
\usepackage{natbib}
\usepackage{booktabs,xcolor}
\usepackage{graphics}
\usepackage{enumitem}
\usepackage{subcaption}
\usepackage{appendix}
\usepackage{adjustbox}
\usetikzlibrary{patterns}
\usepackage{caption}
\usepackage{setspace}
\usepackage[colorlinks=true,linkcolor=blue, urlcolor=blue, citecolor=blue, breaklinks]{hyperref}

\newtheorem{proposition}{Proposition}
\newtheorem{corollary}{Corollary}

\theoremstyle{definition}

\usepackage{etoolbox}
\journal{Economics Letters.}
\patchcmd{\emailauthor}{(#2)}{}{}{}
\patchcmd{\urlauthor}{(#2)}{}{}{}

\usepackage{tikz}

\newcommand{\da}{\mathtt{DA}}

\newcommand{\rsd}{\text{RSD}} 
 
\newcommand{\rk}{\textup{rk}} 

\newcommand{\nee}{\textup{NE}} 
\newcommand{\een}{\textup{EN}}

\begin{document}

\begin{frontmatter}
	\title{The Distribution of Envy in Matching Markets}

	\author[add1]{Josu\'e Ortega}
	\author[add2]{Gabriel Ziegler}
	\author[add3]{R. Pablo Arribillaga}
	\author[add4]{Geng Zhao}
	
	\address[add1]{Queen's University Belfast}
	
	\address[add2]{Freie Universität Berlin and Berlin School of Economics}
	
	\address[add3]{Instituto de Matemática Aplicada San Luis,
Universidad Nacional de San Luis, and CONICET}
	
	\address[add4]{University of California, Berkeley}

		\date{\today}
		
		\begin{abstract}
			We study the distribution of envy in random matching markets under the Deferred Acceptance (DA) algorithm. Using tools from applied probability, we compute the expected number of proposing agents whom nobody envies and those who envy nobody. We obtain an exact finite-market expression for the former, based on a connection with the coupon collector problem, and asymptotic bounds for the latter. To put these quantities into perspective, we compare them to their counterparts under Random Serial Dictatorship (RSD): while RSD assigns a constant fraction of agents to their top choice, both DA and RSD leave exactly $H_n$ proposing agents unenvied in expectation. Our results show that these clearly unimprovable proposing agents constitute a vanishing fraction of the market.
			
		\end{abstract}
		
	\end{frontmatter}
	
	\newpage
	\setcounter{footnote}{0}
\section{Introduction}
\label{sec:intro}

Stable matchings in school choice are defined by the absence of justified envy. 
But how does the lack of justified envy compare to the equally important goal of avoiding envy altogether---whether or not it is justified?\footnote{The absence of envy has long been used as a benchmark of fairness in allocation problems \citep{foley1966resource,varian1974}.}
In this note, we focus particularly on two types of envy relations generated by the most prominent stable mechanism Deferred Acceptance (DA): the students who envy nobody and those who nobody envies. 
Understanding how many students lie in either of these categories is important for three different reasons. 

First, DA tends not to produce a Pareto-efficient matching. Thus, while we know much about the distribution of envy generated by efficient mechanisms (for example, that there is always one student who envies nobody and one student who nobody envies), such conclusions need not carry over to DA. Thus, we know little about the amount of envy generated by DA and its distribution.\footnote{For example, in a serial dictatorship, the first student who picks a school envies nobody; whereas nobody envies the last one to choose. To see that DA may produce an allocation where every student envies someone, we direct the reader to Example 1 in \cite{kesten2010school}.}

Second, students who envy nobody are closely related to those who are assigned to their most preferred alternative, and thus quantifying this top-choice share sheds light on a  welfare measure commonly used by policymakers \citep{dur2018first}. 

And third, understanding how many students lie in each of these two categories is useful to analyze the scope of Pareto-improvements upon the DA baseline. Any student who either envies nobody or is envied by nobody is unimprovable, in the sense that she cannot be made better off by any exchange that starts from the DA outcome \citep{tang2014new}. 
Therefore, the sizes of these sets provide a natural lower bound on the number of all unimprovable students.

We quantify how many students envy nobody or are envied by no one in i.i.d.\ random matching markets.
We derive an exact expression for the expected number of students whom nobody envies, and an asymptotic characterization for the expected number of students who envy nobody. 
Our exact analysis exploits the link between the sequential DA implementation proposed by \cite{mcvitie1971stable} and \cite{wilson1972} and the coupon collector problem. 
Finally we compare our results to their counterparts under Random Serial Dictatorship (RSD) to provide a useful comparison.

\section{Model and Results}
\label{sec:model}

We consider a two-sided matching market with $n$ agents on each side, who have strict preferences over the agents in the opposite side. All agents have complete preferences over all potential matches, i.e.\ the list length is $n$ for everyone. We are only concerned with the welfare of the proposing side of agents, so our model better reflects school choice with students applying to schools, who only have strict priorities over students (but no preferences that need to be taken into account for welfare considerations) rather than men proposing to women, where the welfare concerns both sides of the market.
	
We consider a \emph{random school choice problem} where strict preferences and priorities are drawn independently and uniformly at random, with an equal number of students and schools (and each school's quota is one).
We will use this framework assuming (as the majority of the literature) that the number of students and schools is equal and given by $n$. The benefit of imposing this strong assumption is that it allows us to obtain tractable results. Random matching problems have been useful to understand the average efficiency of several matching mechanisms \cite{wilson1972,knuth1976,pittel1989,liu2016,che2018payoff,che2019efficiency,pycia2019evaluating,nikzad2022rank,ortega2023cost,ortega2025,ortegascwe}, the expected number of stable matchings  \citep{pittel1992likely,immorlica2005}, the potential gains from strategic behaviour \citep{lee2016incentive} and the effect of adding one more agent to a matching problem \citep{ashlagi2017}. 

We denote by $P_n$ a random instance of a school choice problem of size $n$.
First we quantify the fraction of students whom nobody envies. Then we will proceed to do the same for the students who envy nobody.

\subsection{Students whom Nobody Envies}
We use $\nee_n$ to denote the expected number of students assigned to under-demanded schools, or, equivalently, the students whom nobody envies. 
Proposition \ref{thm:expected} shows that the expected number of under-demanded schools equals the $n$-th Harmonic number $H_n:= \sum_{k=1}^n \frac 1k$.

\begin{proposition}
	\label{thm:expected}
	In a random school choice problem of size $n$, $$\nee_n =H_n$$
\end{proposition}

\begin{proof}
	We use an equivalent sequential implementation of Deferred Acceptance due to \citet{mcvitie1971stable} and \cite{wilson1972}. In this implementation, at each step a single currently unmatched student applies to her most preferred school not yet applied to; if rejected, she returns to the queue of unmatched students. McVitie and Wilson show that, for any queueing rule, this sequential algorithm produces exactly the same matching as standard DA.
	
	To analyze the number of under-demanded schools, we follow the principle of deferred decisions \citep{knuth1976} and construct students' preferences on the fly. For each student $i$, generate an infinite sequence $(\sigma_{i,1},\sigma_{i,2},\ldots)$ of i.i.d.\ draws from the uniform distribution over the set of schools $S$. The student's strict preference order is obtained by scanning this sequence and retaining only the first occurrence of each school, deleting repeats. This procedure yields a uniformly random strict preference order. When student $i$ needs to make her next proposal in the sequential DA algorithm, we read forward in her sequence until an untried school is encountered.
	
	Consider now the collection of all raw draws $\{\sigma_{i,t}\}$ that are read during the execution of DA, including those that are discarded because they correspond to repeated schools for a given student. By construction, each such draw is an independent uniform draw from $S$. Moreover, the first time a school $s$ appears among the raw draws that are read, it cannot be discarded (since it is not yet a repeat for that student), and therefore it generates a proposal to $s$. Hence, the time at which every school has appeared at least once among the raw draws is exactly the time at which every school has received at least one proposal. In a one-to-one market with $|I|=|S|=n$, once every school has received a proposal every school is holding some student, so the algorithm must terminate at that time.
	
	A school is under-demanded if and only if it receives proposals from exactly one distinct student during the entire DA execution. The key point is that discarded repeats are irrelevant for this event, for the following reason. Suppose student $i$ generates a draw equal to some previously seen school $x$, so this draw is discarded. Then $i$ must already have proposed to $x$ earlier, and $i$ is drawing again only because she became unmatched after that proposal. This can happen only if some other student proposed to $x$ and displaced $i$ there. In particular, $x$ must have received proposals from at least two distinct students. Therefore, any school that ever appears as a discarded repeat among the raw draws cannot be under-demanded. Equivalently, the under-demanded schools are exactly those schools that appear exactly once among the raw draws read up to termination.
	
	Consequently, the number of under-demanded schools in DA coincides exactly with the number of singleton coupon types in the classical coupon collector problem, evaluated at the stopping time when all coupon types have been collected. \citet{amy} derive the joint distribution of the number of singleton coupons and the stopping time, and show that the expected number of singleton coupons equals the harmonic number $H_n$. It follows that $\nee_n = H_n$, as claimed.
\end{proof}

Since $H_n = \log n + \gamma + o(1)$, Proposition~\ref{thm:expected} implies that the number of unenvied students grows only logarithmically with market size, and that the fraction of unenvied students satisfies $\nee_n/n \to 0$ as $n \to \infty$. Said differently, Proposition \ref{thm:expected} implies that very few students are not envied. For example, for $n=10{,}000$, we expected fewer than 10 unenvied students.

Another interesting feature of Proposition \ref{thm:expected} is that it applies to any positive integer $n$, and in particular to small markets too, unlike the vast majority of the results on random matching markets, which establish asymptotic results. To our knowledge, the only other small markets result is the equivalence amongst strategy-proof and Pareto-efficient mechanisms obtained
by \cite{pycia2019evaluating}.

Additionally, note that the in-degree of a student in the envy graph equals the number of proposals received by her matched school minus one. The arguments in the previous proof imply that a typical school receives $\Theta(\log n)$ proposals (see, e.g., \citealt{pittel1989}). Consequently, typical in-degrees in $G^{\da(P_n)}$ are also of logarithmic order.

\subsection{Students who Envy Nobody}

We use $\een_n$ to denote the expected number of students who envy nobody under DA, i.e.\ those assigned to their most preferred school.
To proceed, we derive an asymptotic approximation to the distribution of students' ranks under DA.

From classical results related to the coupon collector problem \citep{wilson1972, knuth1976, pittel1989}, it is known that when the number of students becomes large, the total number of applications made during the execution of the Deferred Acceptance algorithm is on the order of $n H_n$.\footnote{Here and henceforth, we refer to this asymptotic result and interpretation. For any number of students $n$, it is known that the number of applications lies between $n H_n - O(\log^4 n)$ and $(n-1)H_n+1$ \citep{knuth1976}.} Moreover, the order in which these applications are made is irrelevant for the algorithm's outcome \citep{mcvitie1971stable}. This implies that a typical student makes on the order of $H_n$ applications, and a typical school receives on the order of $H_n$ applications.

Consequently, the probability that a student's application is accepted at any given attempt is asymptotically of order $1/H_n$. This implies that the probability that a student is rejected at her first $k-1$ applications and matched at her $k$-th application is well approximated by
\[
\frac{1}{H_n}\left(1-\frac{1}{H_n}\right)^{k-1}
\]
Thus, the distribution of students' ranks under DA is well approximated by a geometric distribution with parameter $1/H_n$, yielding the following asymptotic approximation.\footnote{\cite{ashlagi2020tiered} establish a similar result to Proposition \ref{prop:geometric}, showing that ranks in a tiered matching model follow a geometric distribution.} \begin{proposition}\label{prop:geometric}
	For sufficiently large $n$ and $k\le n$,
\[
\Pr(\rk_i = k) \approx \frac{1}{H_n}\left(1-\frac{1}{H_n}\right)^{k-1}
\]

\end{proposition}
Because the rank distribution is approximately geometric, the following result is immediate.
\begin{corollary}\label{prop:topchoice}
	For sufficiently large $n$,
	\[
	\een_n \approx \frac{n}{H_n}
	\]
\end{corollary}

Although Corollary \ref{prop:topchoice} is asymptotic and therefore weaker than Proposition~\ref{thm:expected}, which holds exactly for every finite $n$, it remains informative. The number of students who envy nobody is also small, though not as small as the number of unenvied students. For example, when $n=10{,}000$, the approximation in Corollary~\ref{prop:topchoice} predicts that roughly 1{,}100 students envy no one.\footnote{Our proof of Corollary~\ref{prop:topchoice} uses a coupon-collector approximation. Quantifying the associated error terms or convergence rates goes beyond the scope of this note. We expect these errors can be controlled using concentration bounds for occupancy processes. An alternative route to error control is to use a continuum (or hybrid) model, as in \citet{arnostiprobabilistic}.}

\subsection{Simulations}

Simulations confirm that our theoretical predictions are accurate even for relatively small market sizes (Figure~\ref{fig:singletons}).

\begin{figure}[h!]
	\centering
	\includegraphics[width=\textwidth]{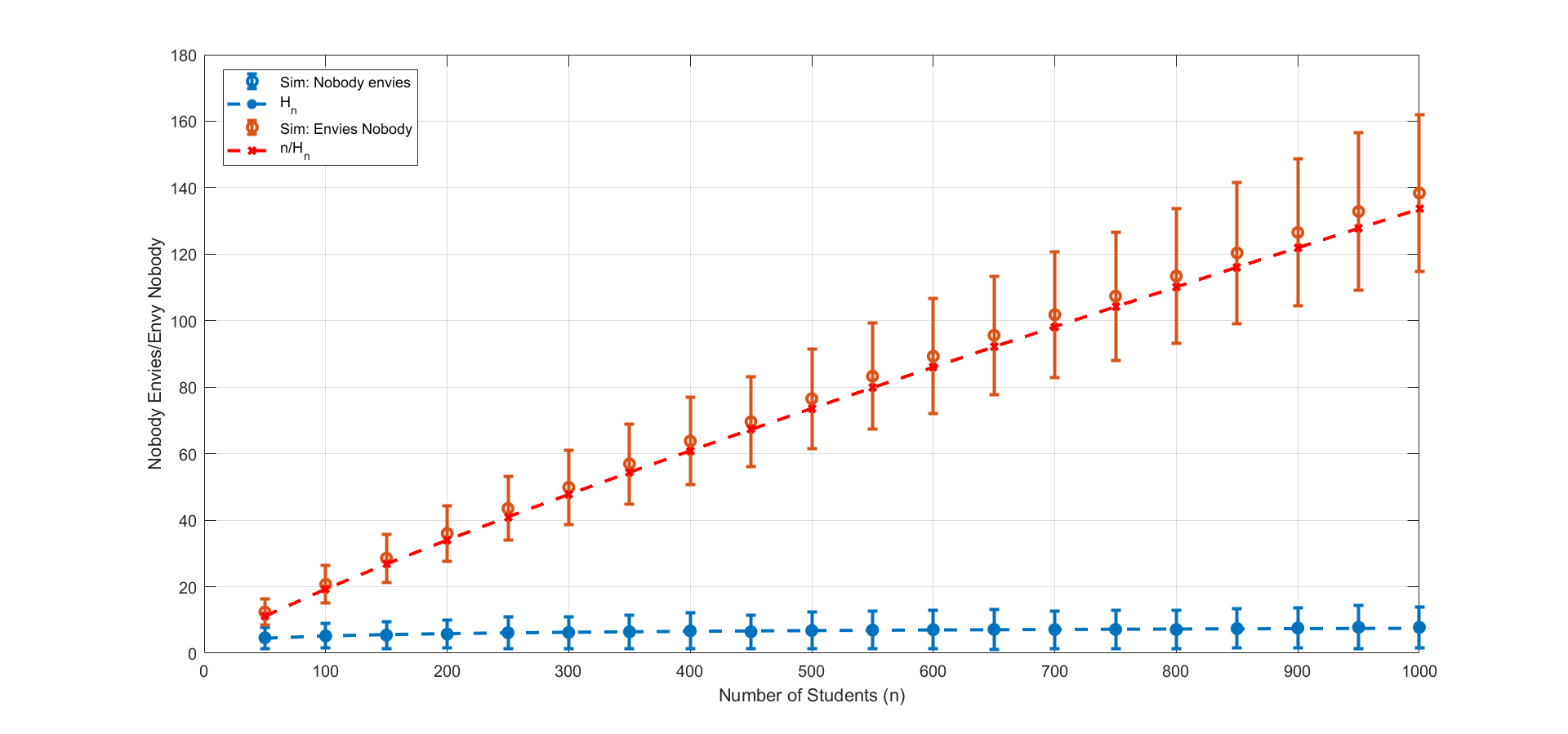} 
	\caption{Average number of students whom nobody envies (blue) or who envy nobody (red) under DA, together with their theoretical predictions (averages over 2{,}000 i.i.d.\ random problems for each $n$).}
	\label{fig:singletons}
\end{figure}

We remark that, while both quantities vanish as a fraction of $n$, they differ sharply in magnitude. Recall that $H_n \approx \log n$ grows extremely slowly. For instance, $H_{100}\approx 5.2$ and $H_{1000}\approx 7.5$. Thus the fraction of unenvied students, $H_n/n$, is tiny for any realistic market size. In contrast, $n/H_n$ remains substantial. Roughly 19 students envy nobody when $n=100$, and 133 when $n=1000$.

\section{Comparison with Random Serial Dictatorship}

A natural question is whether the number of students whom nobody envies or who envy nobody under DA is large or small compared to other mechanisms. We address this by comparing to Random Serial Dictatorship (RSD), where students are ordered uniformly at random and each selects her most preferred available school in turn.

Let $\een_n^\rsd$ denote the expected number of students who envy nobody under RSD. A straightforward calculation shows that $\een_n^\rsd=\frac{n+1}{2}$. To see this, note that the student in position $k$ gets her top choice if and only if it was not taken by the $k-1$ students before her, which occurs with probability $(n-k+1)/n$. Summing over $k$ gives $(n+1)/2$. Thus, unlike DA, RSD matches a constant fraction of students to their top choice.

Let $\nee_n^\rsd$ denote the expected number of students whom nobody envies under RSD. Surprisingly, it exactly equals the corresponding quantity for DA.

\begin{proposition}
	\label{prop:rsd}
	In a random school choice problem of size $n$,
	\[
	\nee_n^\rsd=H_n
	\]
\end{proposition}

\begin{proof}
	Fix a position $k$. Under RSD, a student can only envy an earlier position, so $k$ can only be envied by positions $\ell=k+1,\dots,n$.
	
	Consider $\ell>k$. When student $\ell$ chooses, exactly $\ell-1$ schools have already been taken, so exactly 
	\[
	m:=n-\ell+1
	\]
	schools remain available; call this set $A_\ell$. Let $s_k$ be the school chosen at position $k$.
	
	Student $\ell$ envies position $k$ if and only if she prefers $s_k$ to every school in $A_\ell$, i.e., if and only if $s_k$ is her top-ranked school within the $(m+1)$-set $A_\ell\cup\{s_k\}$. Since preferences are uniform, conditional on the history up to $\ell-1$ each element of $A_\ell\cup\{s_k\}$ is equally likely to be top, hence
	\[
	\Pr(\ell \text{ envies } k \mid \text{history up to }\ell-1)=\frac{1}{m+1}=\frac{1}{n-\ell+2}
	\]
	Therefore
	\[
	\Pr(\ell \text{ does not envy } k \mid \text{history up to }\ell-1)=1-\frac{1}{n-\ell+2}=\frac{n-\ell+1}{n-\ell+2}
	\]
	
	Multiplying these conditional probabilities along $\ell=k+1,\dots,n$ yields
	\[
	\Pr(k \text{ is unenvied})
	=\prod_{\ell=k+1}^{n}\frac{n-\ell+1}{n-\ell+2}
	\]
	
	The product telescopes: the denominator in the factor for $\ell$ cancels the numerator in the factor for $\ell-1$, leaving only the first denominator $n-k+1$ and the last numerator $1$, so
	\[
	\prod_{\ell=k+1}^{n}\frac{n-\ell+1}{n-\ell+2}=\frac{1}{n-k+1}
	\]

	Summing over $k$ and using linearity of expectation gives
	\[
	\mathbb{E}\bigl[\#\{\text{unenvied students}\}\bigr]
	=\sum_{k=1}^{n}\Pr(k \text{ is unenvied})
	=\sum_{k=1}^{n}\frac{1}{n-k+1}
	=H_n
	\]
\end{proof}

Proposition~\ref{prop:rsd} puts our findings into perspective. The number of students receiving their top choice under DA ($\approx n/H_n$) is far below that of RSD ($(n+1)/2$). Yet both mechanisms leave exactly the same expected number of students unenvied: $H_n$. This coincidence is striking given the very different structures of the two mechanisms. Moreover, due to the equivalence between RSD and TTC with random endowments \citep{knuth1996,abdulkadiroglu1998}, this $H_n$ result extends to TTC as well. Table~\ref{tab:comparison} summarizes the comparison.

\begin{table}[h]
	\centering
	\begin{tabular}{lcc}
		\toprule
		& \textbf{DA} & \textbf{RSD} \\
		\midrule
		Students who envy nobody & $\approx n/H_n$ & $(n+1)/2$ \\
		Students whom nobody envies & $H_n$ & $H_n$ \\
		\bottomrule
	\end{tabular}
	\caption{Expected number of envy-free and unenvied students under DA and RSD.}
	\label{tab:comparison}
\end{table}

\section{Discussion}

This paper quantifies two channels through which students are guaranteed to be unimprovable under Deferred Acceptance: being unenvied and envying no one. We show that the expected size of both groups is small, and that their share of the market vanishes as the number of participants grows. Comparing DA to RSD reveals that while DA performs far worse in terms of top-choice assignments, both mechanisms yield exactly $H_n$ unenvied students---suggesting this quantity reflects a fundamental property of random markets rather than a deficiency of any particular mechanism.

Several questions remain open. Does the $H_n$ result extend to other mechanisms? What happens in many-to-one markets or when preferences are correlated? Our simulations indicate that unimprovable students remain rare in many-to-one settings and become even fewer when preferences are correlated. Finally, can the fraction of unimprovable students be characterized beyond the lower bound provided by unenvied and envy-free students? We confirm that this fraction converges to zero in a companion paper \citep{ortega2025}.

\section*{Acknowledgements} We are grateful to an anonymous reviewer and the Editor for very thoughtful comments which significantly improved our paper.

\newpage
	\singlespacing      
\setlength{\parskip}{-0.3em} 


\end{document}